\documentclass{article}
\usepackage{arxiv}

\usepackage{amsthm}

\usepackage{xcolor}  % Needed for colours - load before mdframed

\usepackage{amssymb}
\usepackage{amsmath}
\usepackage{mathdots}

\usepackage{mathtools}

\usepackage{tensor}

\usepackage{urwchancal}
\DeclareFontFamily{OT1}{pzc}{}
\DeclareFontShape{OT1}{pzc}{m}{it}{<-> s * [1.10] pzcmi7t}{}
\DeclareMathAlphabet{\mathpzc}{OT1}{pzc}{m}{it}

\usepackage{graphicx}
\usepackage{ifpdf}
\ifpdf
\usepackage{epstopdf}
\PrependGraphicsExtensions{.svg}
\fi
\newtheorem{theorem}{Theorem}[section]
\newtheorem{lemma}[theorem]{Lemma}

\newtheorem{proposition}[theorem]{Proposition}

\providecommand{\R}{\mathbb{R}}

\providecommand{\SO}{\mathbf{SO}}

\providecommand{\GL}{\mathbf{GL}}
\providecommand{\SE}{\mathbf{SE}}

\providecommand{\grpG}{\mathbf{G}}

\providecommand{\gothg}{\mathfrak{g}}

\providecommand{\gothX}{\mathfrak{X}} % as in X(M)

\providecommand{\so}{\mathfrak{so}}

\providecommand{\calE}{\mathcal{E}}

\providecommand{\calL}{\mathcal{L}}
\providecommand{\calM}{\mathcal{M}}
\providecommand{\calN}{\mathcal{N}}

\providecommand{\calP}{\mathcal{P}}

\providecommand{\vecL}{\mathbb{L}}

\providecommand{\tT}{\mathrm{T}} % tangent maps eg $\tT \calM$
\providecommand{\Id}{I} % identity of a matrix group.

\DeclareMathOperator{\tr}{tr}

\DeclareMathOperator{\diag}{diag}

\DeclareMathOperator{\Ad}{Ad}
\DeclareMathOperator{\ad}{ad}

\providecommand{\td}{\mathrm{d}}
\providecommand{\tD}{\mathrm{D}}
\providecommand{\tL}{\mathrm{L}}

\providecommand{\ddt}{\frac{\td}{\td t}}

\usepackage{accents}
\makeatletter
\providecommand{\scirc}{%
    \hbox{\fontfamily{\rmdefault}\fontsize{0.4\dimexpr(\f@size pt)}{0}\selectfont{\raisebox{-0.52ex}[0ex][-0.52ex]{$\circ$}}}}
\makeatother
\mathchardef\mhyphen="2D
\usepackage{graphicx}      % include this line if your document contains figures
\usepackage{natbib}        % required for bibliography
\usepackage{hyperref}

\newcommand{\publicationdetails}
{Under review}
\newcommand{\publicationversion}
{Preprint}
\begin{document}

\title{Exploiting Equivariance in the Design of Tracking Controllers for Euler-Poincare Systems on Matrix Lie Groups}
\headertitle{Exploiting Equivariance in the Design of Tracking Controllers for Euler-Poincare Systems on Matrix Lie Groups}

\author{
\href{https://orcid.org/0000-0001-2345-6789}{\includegraphics[scale=0.06]{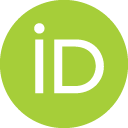}\hspace{1mm}
Matthew Hampsey}
\\
	Systems Theory and Robotics Group \\
	Australian National University \\
    ACT, 2601, Australia \\
	\texttt{matthew.hampsey@anu.edu.au} \\
	\And	\href{https://orcid.org/0000-0001-2345-6789}{\includegraphics[scale=0.06]{orcid.png}\hspace{1mm}
Pieter van Goor}
\\
	Systems Theory and Robotics Group \\
	Australian National University \\
    ACT, 2601, Australia \\
	\texttt{pieter.vangoor@anu.edu.au} \\
    \And	\href{https://orcid.org/0000-0001-2345-6789}{\includegraphics[scale=0.06]{orcid.png}\hspace{1mm}
Ravi Banavar}
    \\
        Interdisciplinary Programme in Systems and\\ Control Engineering \\
        Indian Institute of Technology Bombay \\
        India \\
        \texttt{banavar@iitb.ac.in} \\
\And	\href{https://orcid.org/0000-0001-2345-6789}{\includegraphics[scale=0.06]{orcid.png}\hspace{1mm}
Robert Mahony}
\\
	Systems Theory and Robotics Group \\
	Australian National University \\
	ACT, 2601, Australia \\
	\texttt{robert.mahony@anu.edu.au} \\
}

\maketitle

\thispagestyle{empty}
\pagestyle{empty}

\begin{abstract}                % Abstract of not more than 250 words.
The trajectory tracking problem is a fundamental control task in the study of mechanical systems.
A key construction in tracking control is the error or difference between an actual and desired trajectory.
This construction also lies at the heart of observer design and recent advances in the study of equivariant systems have provided a template for global error construction that exploits the symmetry structure of a group action if such a structure exists.
Hamiltonian systems are posed on the cotangent bundle of configuration space of a mechanical system and symmetries for the full cotangent bundle are not commonly used in geometric control theory.
In this paper, we propose a group structure on the cotangent bundle of a Lie-group and leverage this to define momentum and configuration errors for trajectory tracking drawing on recent work on equivariant observer design.
We show that this error definition leads to error dynamics that are themselves ``Euler-Poincare like'' and use these to derive simple, almost global trajectory tracking control for fully-actuated Euler-Poincare systems on a Lie-group state space.
\end{abstract}

%
%       We show that any smooth trajectory for such a system admits an intrinsic global error, the dynamics of which can be linearised within a single coordinate chart.
%       We use this construction and the resultant local coordinates for LQR tracking control, a framework we term the Equivariant Regulator (EqR).

%===============================================================================
%%%%%%%%%%%%%%%%%%%%%%%%%%%%%%%%%%%%%%%%%%%%%%%%%%%%%%%%%%%%%%%%%%%%%%%%%%%%%%%%%%%%%%%%%%%%%%%%%%
\section{Introduction}
%%%%%%%%%%%%%%%%%%%%%%%%%%%%%%%%%%%%%%%%%%%%%%%%%%%%%%%%%%%%%%%%%%%%%%%%%%%%%%%%%%%%%%%%%%%%%%%%%%
The proliferation of small autonomous systems in modern society has emphasised the need for high performance tracking control for mechanical systems.
%Models that exploit the underlying geometry of the mechanical system have been shown to lead to high performance and robust control.
There is a rich history of control design for this problem.
Early approaches largely focused on the control of robotic manipulators (\cite{SLOTINE1989509,10.1115/1.3139651}) leading to a range of robust and adaptive controllers that are now well established in the literature.
Passivity-based control of mechanical systems was developed in the 90s and into the 2000s.
The landmark paper (\cite{schaftL2gainPassivityTechniques1996}) exploited the passivity property for the stabilisation of Hamiltonian systems.
In (\cite{ortegaEnergyshapingPortcontrolledHamiltonian1999}), passivity-based control is extended from mechanical systems to more general port-Hamiltonian systems with dissipation.
This result was extended by \cite{fujimotoTrajectoryTrackingControl2003} to the more general tracking problem for mechanical and port-Hamiltonian systems by the choice of canonical coordinate transformation.
A key insight in this work is the selection of local error coordinates that are themselves Hamiltonian and also preserve required passivity conditions.
In \cite{mahonyNovelPassivitybasedTrajectory2019}, a passivity-based tracking control for general conservative mechanical systems was proposed where the reference trajectory is used to parameterise a single-coordinate Hamiltonian system that compensates for energy in the tracking error.
All these control designs depend explicitly on a choice of generalised coordinates in which the Euler-Lagrange equations or the equivalent Hamiltonian system on $\R^{2n}$ can be locally formulated.

Parallel to the above development, there has also been considerable work on coordinate-free control design for mechanical systems.
These works require additional geometric structure to be present: either a symmetry group acting on the state-space (typically the system is posed directly on a Lie-group) or the presence of a metric.
In (\cite{blochControlledLagrangiansStabilization2000}), the method of controlled Lagrangians is introduced in order to asymptotically stabilise mechanical systems with symmetry.
In (\cite{bulloTrackingFullyActuated1999a}), given an error function acting on a manifold and a Riemannian metric, the associated transport map is used to construct velocity error and design a stabilising control.
Later, in the 2010's, the consumerisation of UAV technology and the desire to control these vehicles through a wide range of flight modalities called for an intrinsic approach to control.
In (\cite{cabecinhasOutputfeedbackControlAlmost2008}, \cite{leeGeometricTrackingControl2010}), almost-global approaches to quadrotor tracking control are proposed by utilising the natural Lie group structure present on $\SO(3)$.

More recently, there have been significant advances in observer design for systems with symmetry, motivated largely by the requirement for state estimation for aerial robotics.
In (\cite{mahonyNonlinearComplementaryFilters2008}), the $\SO(3)$ symmetry is used to develop an almost-globally stable attitude observer.
In (\cite{barrauInvariantExtendedKalman2016}), the Lie group $\SE_2(3)$ is introduced for inertial navigation filtering and it is shown that for ``group-affine'' systems, the correct choice of symmetry leads to an exact linearisation of the error dynamics.
In (\cite{vangoorEquivariantFilterEqF2023}), a filter design for systems with transitive symmetries is proposed.
Equivariance of the system output is shown to lead to reduced linearisation error.
A common theme in these approaches is the fact that choosing the \textit{correct} symmetry, motivated by equivariance, appears to lead simplified error dynamics and improved observer performance.

In this paper,  we consider the trajectory tracking control problem for fully-actuated mechanical systems whose configuration space is a matrix Lie group.
In this case, for constant inertia, the system dynamics are given by the Euler-Poincare equations (\cite{holmEulerPoincareEquations1998}).
The key to our approach is extending the natural symmetry on the Lie-group configuration space to a transitive symmetry on the system phase space associated with a particular semi-direct-product Lie-group  (\cite{jayaramanBlackScholesTheoryDiffusion2020}, \cite{holmEulerPoincareEquations1998}).
We observe that the Euler-Poincare equations are a sub-behaviour of a larger system where the velocity input is a free input.
For this non-physical system we define an input action, acting on velocity and force inputs, and show that the extended Euler-Poincare equations are equivariant with respect to the group action defined earlier.
This structure allows us to employ the tools developed in the equivariant observer field (\cite{mahonyEquivariantSystemsTheory2020}, \cite{geEquivariantFilterDesign2022}, \cite{vangoorEquivariantFilterEqF2023}) to define an equivariant tracking error.
The resultant error dynamics are nice and we use these error coordinates to define a Lyapunov function.
Using the true velocity as the free velocity input recovers the physical system dynamics and the torque input can now be designed to ensure a global decrease condition on the Lyapunov function leading to a novel trajectory tracking controller design.

%-----------------------------------------------------------------
\section{Preliminaries}
%-----------------------------------------------------------------
In this section we introduce some notation and mathematical identities that will be used throughout the paper.
For a detailed overview, we refer the reader to (\cite{leeSmoothManifolds2003}).

\subsection{Notation}

%Let $\vecV$ and $\vecW$ be vector spaces. For a function $g : \vecV \to \vecW$, the notation $g[v]$ will be used to indicate that $g$ is a linear function of $\vecV$.

Let $\calM$ and $\calN$ denote smooth manifolds.
For an arbitrary point $Q \in \calM$, denote the tangent space of $\calM$ at $Q$ by $\tT_{Q}\calM$.
For a smooth function $h : \calM \to \calN$ the notation
\begin{align*}
    \tD _{Q|\zeta} h(Q):\tT_{\zeta}\calM\to \tT_{h(\zeta)}\calN \\
    \delta \mapsto \tD_{Q|\zeta}h(Q)[\delta]
\end{align*}
denotes the differential of $h(Q)$ evaluated at $Q = \zeta$ in the direction $\delta \in \tT_{\zeta}\calM$.
When the basepoint and argument are implied the notation $\tD h$ will also be used for simplicity.

For a function with multiple arguments $h : \calM \times \calN \to \calP$,
the notation $\tD_1 h : \tT \calM\to \tT \calP$ and $\tD_2 h: \tT \calN \to \tT \calP$
are used to denote the derivatives with respect to the first and second argument, respectively.
The space of smooth vector fields on $\calM$ is denoted with $\gothX(\calM)$.

Let $\grpG$ denote an m-dimensional matrix Lie group and denote the identity element with $\Id$.
The \emph{Lie algebra}, $\gothg$ of $\grpG$, is identified with the tangent space of $\grpG$ at identity, $\gothg\simeq \tT_{\Id}\grpG  \subset \R^{m \times m}$.
Given arbitrary $X \in \grpG$, \emph{left translation by $X$} is defined by
\mbox{$\mathrm L_X : \grpG \to \grpG$}, \mbox{$\mathrm L_X(Y) = XY.$}
This induces a corresponding left translation on $\gothg$,
\mbox{$\tD \mathrm L_X : \gothg \to \tT_X \grpG,$}
defined by $\tD \mathrm L_{X} U = XU$.
Similarly, \emph{right translation by $X$} is defined by
\mbox{$\mathrm R_X : \grpG \to \grpG$}, \mbox{$\mathrm R_X(Y) = YX.$}
This also induces a corresponding right translation on $\gothg$,
\mbox{$\tD \mathrm R_X : \gothg \to \tT_X \grpG$}, defined by $\tD \mathrm R_{X}U = UX$.
Given $X \in \grpG$, the adjoint map \mbox{$\Ad_X: \gothg \to \gothg$} is defined by \mbox{$\Ad_X(U) = X U X^{-1}$}.
Given $u \in \gothg$, the ``little'' adjoint map \mbox{$\ad_u: \gothg \to \gothg$} is defined by \mbox{$\ad_u v = uv - vu$}.

Let $\grpG$ be a Lie group and $\calM$ a smooth manifold.
A right action is a smooth function $\phi: \grpG \times \calM \to \calM$ satisfying the \emph{identity} and \emph{compatibility } properties:
\begin{align*}
    \phi(\Id, Q)        & = Q           \\
    \phi(Y, \phi(X, Q)) & = \phi(XY, Q)
\end{align*}
for all $Q \in \calM$ and $X \in \grpG$.

For a vector space $V$, let $V^\ast$ denote the dual space of $V$; that is, the set of linear functionals acting on $V$.
If $V, W$ are vector spaces and $A : V \to W$ is a linear map, then $A^\ast$ denotes the dual map $A^\ast : W^\ast \to V^\ast$ defined by
\begin{align*}
    (A^\ast p)(v) = p(Av),
\end{align*}
for all $p \in W^\ast$, $v \in V$.

The Frobenius inner product is an inner product on $\R^{m \times m}$, defined by
\begin{align*}
        \langle A, B \rangle = \tr(A^\top B).
\end{align*}
Note that $\langle A, BC \rangle= \langle B^\top A, C \rangle$.
Given an inner product space $W$ and a subspace $V \subseteq W$, let $V^\perp$ denote the orthogonal subspace, defined by 
\begin{align*}
    V^\perp = \left\{ w \in W : \langle v, w \rangle = 0 \; \forall v \in V  \right\}.    
\end{align*}
Every vector $w \in W$ can be expressed as a sum $w = u + v$, with $v \in V$ and $u \in V^\perp$ (\cite{axler1997linear}).
Then the subspace projection operator $\mathbb{P}_V$ can be defined by $\mathbb{P}_V w = v$.

\subsection{Useful identities}

The following identities will be useful later in this document; we state them without proof.
If $X = X(t)$ is a trajectory on a matrix Lie group $\grpG$ with derivative $\dot{X} = X U$, with $U = U(t) \in \gothg$ then one can compute the derivative
\begin{align}
    \ddt \Ad_{X} &= \Ad_{X} \ad_{U} = \ad_{\Ad_{X}U} \Ad_{X}  \label{eq:Ad_deriv}.
 \end{align}
 An immediate corollary is that one also has the following identities
 \begin{align}
    \ddt &\Ad_{X}^\ast  = \ad_{U}^\ast \Ad_{X}^\ast = \Ad_{X}^\ast \ad_{\Ad_{X}U}^\ast  \label{eq:Ad_star_deriv}\\
   \ddt &\Ad_{X^{-1}}  = - \ad_{U} \Ad_{X^{-1}} = - \Ad_{X^{-1}} \ad_{\Ad_{X}U} \label{eq:Ad_inv_deriv}\\
    \ddt &\Ad^\ast_{X^{-1}} = -\Ad_{X^{-1}}^\ast \ad_U^\ast = - \ad_{\Ad_{X}U}^\ast  \Ad_{X^{-1}}^\ast. \label{eq:Ad_inv_star_deriv}
\end{align}

\subsection{The Lie Group $\SO(3)$}
%-----------------------------------------------------------------

The Lie group $\SO(3)$ is defined by the set of matrices
\begin{align*}
   \SO(3) = \left\{ R \in \GL(3) : R^\top R = R R^\top = \Id_3, \det(R) = 1 \right\}.
\end{align*}
The associated Lie algebra $\so(3)$ is defined by
\begin{align*}
   \so(3) = \left\{ \Omega^\times \in \R^{3 \times 3}: \Omega \in \R^3 \right\},
\end{align*}
where the skew-symmetric map $( \cdot )^\times : \R^3 \to \so(3)$ maps
\begin{align*}
    \begin{pmatrix}
        v_1 \\ v_2 \\ v_3
    \end{pmatrix} \mapsto \begin{pmatrix}
    0 & -v_3 & v_2\\
    v_2 & 0 & -v_1\\
    -v_2 & v_1 & 0
\end{pmatrix}.
\end{align*}

There is an isomorphism between $\so(3)$ and $\R^3$ by identifying $\Omega^\times \cong \Omega$.
Denote the corresponding inverse map by
\begin{align*}
    (\cdot)^\vee : \so(3) \to \R^3.
\end{align*}
The set of linear functionals $\so^*(3) : \so(3) \to \R$ can be identified with $\R^3$ by identifying $P \in \so^\ast(3)$ with the vector $p \in \R^3$ such that $P(v^\times) = p^\top v$ for all $v \in \R^3$.
Using the fact that $\Ad_R \Omega^\times = (R \Omega)^\times$ and $(v \times u)^\times = v^\times u^\times - u^\times v^\times$, we also have the following identifications (\cite{marsden1994introduction})
\begin{align}
  \Ad_R \Omega^\times &\cong R \Omega, & \Ad_R^\ast P &\cong R^\top p \notag\\
  \ad_{v} u &\cong v \times u, & \ad_{v}^\ast P &\cong - v \times p \label{eq:r3_ident}.
\end{align}
The skew-symmetric component of a matrix $A \in \R^{3 \times 3}$ is given by the projection $\mathbb{P}_{\so(3)}(A)$, which can be explicitly computed to be given by
\begin{align}
    \mathbb{P}_{\so(3)}(A) = \frac{1}{2}(A - A^\top) \label{eq:so3_proj}.
\end{align}

%-----------------------------------------------------------------
\section{Group actions and Equivariance for Extended Euler-Poincare Equations}
%-----------------------------------------------------------------
%-----------------------------------------------------------------
\subsection{Extended Euler-Poincare equations}
%-----------------------------------------------------------------

Let $\grpG$ be a matrix Lie group and let $\tT^\ast \grpG$ denote the cotangent bundle of $\grpG$.
Let $Q \in \grpG$ and $P \in \gothg^\ast$ denote elements of the group and the Lie coalgebra, respectively.
The cotangent bundle can be left trivialised by the correspondence
$\tD \tL_{Q}^\ast : \tT^\ast_Q \grpG \to \gothg^\ast$.
Thus, an element $(Q,P) \in \grpG \times \gothg^\ast$ corresponds to an element $(Q, \tD \tL_{Q^{-1}}^\ast [P]) \in \tT^\ast_Q \grpG$.

Consider the \emph{system}
\begin{subequations}\label{eq:dynamics}
\begin{align}
\dot{Q} &= Q U \\
\dot{P} &= \ad_U^\ast P + \tau,
\end{align}
\end{subequations}
defined for state  $(Q, P) \in \grpG \times \gothg^\ast$ and inputs $U, \tau \in \vecL \coloneqq \gothg \times \gothg^\ast$.
Define a system function
$f : \vecL \to \gothX (\grpG \times \gothg^\ast)$ by
\begin{align}
f_{(U,\tau)}(Q,P) = f((Q,P),(U,\tau)) := (Q U , \ad_U^\ast P + \tau) \label{eq:f_def}
\end{align}
such that \eqref{eq:dynamics} can be expressed as $(\dot{Q},\dot{P}) = f_{(U,\tau)}(Q,P)$.
We will term Equations \eqref{eq:dynamics} the \emph{extended Euler-Poincare} equations as discussed below.

The extended Euler-Poincare equations \eqref{eq:dynamics} are closely related to the classical Euler-Poincare equations.
Let $\calL : \tT\grpG \rightarrow \R$ be a left-invariant Lagrangian and let $\ell: \gothg \to \R$ be its restriction to $\gothg$.
Consider the constraint $P = \tD \ell (U)  \in \gothg^\ast$ applied to the free input $U \in \gothg$ in \eqref{eq:dynamics}.
With this constraint then \eqref{eq:dynamics} are the (forced) Euler-Poincare equations (\cite{holmEulerPoincareEquations1998}).
That is, the physical dynamics associated with the Lagrangian $\calL$ forced by $\tau$ is a sub-behaviour of the general dynamics encoded by \eqref{eq:dynamics} obtained when the additional constraint $P = \tD \ell (U)$ is enforced.
The more general dynamics defined by \eqref{eq:dynamics} are important since the equivariant theory that we employ is agnostic to the physical constraint between momentum and velocity.

In practice, in this paper we will use a quadratic left trivialised Lagrangian $\ell = \langle U, \mathbb{I} U \rangle$ where $\mathbb{I} : \gothg \to \gothg^*$ is a positive definite operator and $\langle \cdot, \cdot \rangle$ is the natural contraction.
We ensure that only physical trajectories are considered by choosing $U = \mathbb{I}^{-1} P $ as the velocity input at the appropriate point in the control design.
%In this paper, we will refer to \eqref{eq:dynamics} as the ``extended Euler-Poincare'' equations.
%, even though in the literature the term Euleris generally reserved for these equations with the constrained momentum.

%\todo{The Fiber derivative for $l$ can be defined in the obvious manner: $\mathbb{F}l(u) [w]= \mathbb{F}L(I, u) [w]$.
%Then it is clear that $\frac{\partial l}{\partial u} = \mathbb{F}l(u)$ and so $\frac{\partial l}{\partial u}$ is the momentum corresponding to the velocity $u$.
%The corresponding forced equation is given by
%\begin{align*}
%    (\frac{\td}{\td t} - \ad_u^\ast) \frac{\partial l}{\partial u} = \tau,
%\end{align*}
%with $\tau \in \tT^\ast \grpG$.
%

% Now, consider the following system on $\grpG \times \gothg^\ast$ with input space $\vecL = \gothg \times \gothg^\ast$:
% \begin{subequations}
% \begin{align}
% \dot{Q} &= QU\\
% \dot{P} &= \ad_U^\ast P + \tau,
% \end{align} \label{eq:dynamics}
% \end{subequations}
% where $(Q, P) \in \grpG \times \gothg^\ast$ and $(U, \tau) \in \vecL$.

\subsection{Symmetry and group action}
\label{sec:symmetry}

%Let $\grpG$ be a Lie group and let $\gothg^\ast$ be the dual of its Lie algebra.
We wish to study $\grpG \times \gothg^\ast$ as a group in its own right.
%There are several natural ways of extending the group structure of $\grpG$ to $\grpG \times \gothg^\ast $.
We choose to imbue $\grpG \times \gothg^\ast$ with the semidirect group structure $\grpG^\ast_{\ltimes} \coloneqq \grpG \ltimes \gothg^\ast$ given by the group product (\cite{holmEulerPoincareEquations1998,jayaramanBlackScholesTheoryDiffusion2020})
\begin{align}
    (Q_1, P_1)(Q_2, P_2) = (Q_1 Q_2, \Ad_{Q_2}^\ast P_1 + P_2) \label{eq:semidirect},
\end{align}
for $(Q_1, P_1), (Q_2, P_2) \in \grpG^\ast_{\ltimes}$.
One can verify that the inverse  $(X, P)^{-1}$ is given by $(X^{-1}, -\Ad_{X^{-1}}^\ast P)$.

Although strictly speaking the system state space is $\grpG \times \gothg^\ast$ we will identify the state-space with the group and write $(Q,P) \in \grpG^\ast_\ltimes$.
With this identification a natural right group action $\phi : \grpG^\ast_{\ltimes} \times \grpG^\ast_{\ltimes} \to \grpG^\ast_{\ltimes}$ is given by canonical right multiplication
\begin{align}
    \phi((Q_1, P_1), (Q_2, P_2)) &\coloneqq (Q_2, P_2)(Q_1, P_1) \notag\\
                                 &= (Q_2 Q_1, \Ad_{Q_1}^\ast P_2 + P_1)\notag \\
                                 &= (\phi^Q, \phi^P), \label{eq:phi_def}
\end{align}
where we have defined the component maps
% the $\grpG$ and $\gothg^\ast$ component of this map by $\phi^Q((Q_1, P_1), (Q_2, P_2))$ and $\phi^P((Q_1, P_1), (Q_2, P_2))$, respectively, so that
\begin{align*}
    \phi^Q((Q_1, P_1), (Q_2, P_2)) &\coloneqq Q_2 Q_1 \in \grpG\\
    \phi^P((Q_1, P_1), (Q_2, P_2)) &\coloneqq \Ad_{Q_1}^\ast P_2 + P_1 \in \gothg^\ast.
\end{align*}

We also consider a group action on the input space $\psi: \grpG^\ast_{\ltimes} \times \vecL \to \vecL$ defined by
\begin{align}
    \psi((X,& P), (U, \tau)) \notag\\
        &\quad = \left( \psi^U \left( (X, P), (U, \tau) \right), \psi^\tau\left((X, P), (U, \tau) \right) \right) \notag\\
        &\quad \coloneqq \left( \Ad_{X^{-1}} U, \quad \Ad_{X}^\ast \tau - \ad_{\Ad_{X^{-1}} U}^\ast P \right). \label{eq:psi_def}
\end{align}
The first component $\psi^U$ is like a coordinate transformation of the velocity $U \in \gothg$ from reference to body coordinates.
The first part of the second component $\psi^\tau$ plays a similar role while the second part contains a coupling term closely related to coriolis force.
There is no requirement that the new velocity $\Ad_{X^{-1}} U$ is related to momentum in any way.
This is a purely algebraic construction associated with equivariance of the extended Euler-Poincare equations as an abstract system.

Note that not only is the group action $\phi$ different from the original canonical left action associated with the classical Euler-Poincare equations, but the symmetry captures the geometry of the controlled system on the full cotangent bundle.
The key observation is that the original symmetry models invariance of the Lagrangian of a physical system whereas the new geometry models equivariance of the extended Euler-Poincare equations.
%This is clear in the control input action that acts on both the generalised force and the velocity input that is only free for the extended Euler-Poincare equations.

%-----------------------------------------------------------------
\subsection{Equivariance}
\label{sec:equivariance}
%-----------------------------------------------------------------

A key property of the symmetry \eqref{eq:phi_def} and \eqref{eq:psi_def} is that there is a natural equivariant structure when one considers the dynamics \eqref{eq:dynamics}.
This will be important because it implies the existence of error dynamics with the same structure as $f$ (\cite{mahonyEquivariantSystemsTheory2020}).

\begin{lemma}
\label{lem:lemma_eq_d2}
Given the semidirect product group $\grpG^\ast_{\ltimes}$, consider the system function \eqref{eq:f_def} that assigns the vector field.
Then $f$ is equivariant with respect to the group action $\phi$ \eqref{eq:phi_def} and the input action $\psi$ \eqref{eq:psi_def} in the sense that
\begin{align*}
\tD \phi_{(X_Q,X_P)}f (Q,P) & = f((Q,P),(U,\tau))  \\
& = f(\phi_{(X_Q,X_P)}(Q,P),\psi_{(X_Q,X_P)}(U,\tau)),
\end{align*}
for all $X = (X_Q, X_P) \in \grpG_\ltimes^\ast$.
\end{lemma}

\begin{proof}
    The differential $\tD \phi_{(X_Q,X_P)}$ in the direction $(Q V, W)$ can be computed to be
    \begin{align*}
        \tD \phi_{(X_Q, X_P)}& (Q, P)[Q V, W] \\
                & = \frac{\td}{\td s}\bigg\rvert_{s=0} \phi((X_Q, X_P), (Q e^{Vs}, Ws + P))\\
                             &= (Q V X_Q, \Ad_{X_Q}^\ast  W).
    \end{align*}

    As such, $\tD\phi$ transforms the dynamics $f$ by
\begin{align}
\tD &\phi_{(X_Q, P)}(Q, P) \left[  f((Q, P), (U, \tau)) \right]\notag \\
                    %&= \tD \phi_{(X_Q, P)}(Q, P) \left[ Q U, \ad_U^\ast  P + \tau  \right]\notag\\
                    &= (Q U X_Q, \Ad_{X_Q}^\ast  ( \ad_U^\ast  P + \tau)) \notag\\
                    %&= (Q X_Q X_Q^{-1} U X_Q, \Ad_{X_Q}^\ast  \tau +  \ad_{\Ad_{X_Q}^{-1}U}^\ast  \Ad_{X_Q}^\ast  P \notag\\
                    %&= (Q X_Q \Ad_{X_Q^{-1}} U , \Ad_{X_Q}^\ast  \tau - \ad_{\Ad_{X_Q^{-1}} U}^\ast  X_P\\
                    %& \qquad \qquad + \ad_{\Ad_{X_Q^{-1}} U}^\ast  X_P - \ad_{\Ad_{X_Q^{-1}} U}^\ast  X_P) \\
                    &= (Q X_Q \Ad_{X_Q^{-1}} U, ( \Ad_{X_Q}^\ast  \tau - \ad_{\Ad_{X_Q^{-1}} U}^\ast  X_P) \notag\\
                    & \qquad \qquad + \ad_{\Ad_{X_Q^{-1}} U}^\ast  (\Ad_{X_Q}^\ast  P + X_P)) \label{eq:equiv_1}\\
                    &= (\phi^Q((X_Q, X_P), (Q, P)) \psi^U((X_Q, X_P), (U, \tau)), \notag\\
                    & \quad \quad  \ad_{\psi^U((X_Q, X_P), (U, \tau))}^\ast  \phi^P((X_Q, X_P), (Q, P))) \notag\\
                    & \quad \quad + \psi^\tau((X_Q, X_P), \tau )\label{eq:equiv_2}\\
                    &= f(\phi((X_Q, X_P), (Q, P)), \psi((X_Q, X_P), (U, \tau))) \label{eq:equiv_3},
                    %&= (\phi^Q(A, X) \psi^U(A, u), -\ad_{\psi^U(A, u)}^\ast  ( \Ad_{A^{-1}}^\ast  ( P - b)) + \Ad_{A^{-1}}^\ast  \tau - -\ad_{\psi^U(A, u)}^\ast  ( \Ad_{A^{-1}}^\ast  b)
\end{align}
where \eqref{eq:equiv_1} follows from \eqref{eq:Ad_star_deriv}, \eqref{eq:equiv_2} follows from the definitions of $\psi$ and $\phi$ and \eqref{eq:equiv_3} follows from the definition of $f$.
\end{proof}

\section{Control Error}
%-----------------------------------------------------------------

The main contribution of the paper is the definition of a novel globally defined tracking error.
In this section, we exploit the symmetry introduced in Section~\ref{sec:symmetry} to show that corresponding error dynamics are ``nice'' in the sense that they are aso extended Euler-Poincare equations.

Let $(Q(t), P(t)), (Q_d(t), P_d(t)) \in  \grpG^\ast _{\ltimes}$ be the system and desired trajectory, respectively, written as trajectories on $\grpG^\ast _{\ltimes}$.
Define the error trajectory to be
\begin{align}
    e &= \phi((Q_d, P_d)^{-1}, (Q, P)) \notag\\
      %&= \phi((Q_d^{-1}, -\Ad_{Q_d^{-1}}^\ast  P_d), (Q, P))\\
      %&= (Q, P)(Q_d^{-1}, -\Ad_{Q_d^{-1}}^\ast  P_d)\\
      &= (Q Q_d^{-1}, \Ad_{Q_d^{-1}}^\ast  (P - P_d)) =: (Q_E, P_E) \in \grpG^\ast _{\ltimes}. \label{eq:e_def}
\end{align}
In contrast to the standard formulation of tracking problems for mechanical systems where error is typically formulated as a configuration and velocity error, the error variables in our study evolve on the configuration and momentum domains.

% We have
% \begin{align*}
%     \dot{Q} &= Q U\\
%     \dot{P} &= \ad_{U}^\ast  P + \tau
% \end{align*}
% and
% \begin{align*}
%     \dot{Q}_d &= Q_d U_d\\
%     \dot{P}_d &= \ad_{U_d}^\ast  P_d + \tau_d.
% \end{align*}

\begin{theorem}
    \label{thm:error_dyn}
Consider trajectories $(Q(t), P(t))$ and $(Q_d(t)$, $P_d(t))$ \eqref{eq:dynamics} for inputs $(U,\tau)$ and $(U_d,\tau_d)$ respectively.
That is
\begin{align*}
 \dot{Q} = Q U  \quad\quad&\text{and}\quad\quad \dot{P} = \ad_{U}^\ast  P + \tau, \\
 \dot{Q}_d = Q_d U_d  \quad\quad&\text{and}\quad\quad
 \dot{P}_d = \ad_{U_d}^\ast  P_d + \tau_d.
\end{align*}
Define $e = \phi((Q_d, P_d)^{-1}, (Q, P))$ to be the error trajectory \eqref{eq:e_def}.

Let $\psi : \grpG^\ast _{\ltimes} \times \vecL \to \vecL$ be given by \eqref{eq:psi_def}.
Define the \emph{control errors}
\begin{subequations}\label{eq:control errors}
\begin{align}
    \tilde{U}_\psi \coloneqq \psi^U((Q_d, P_d)^{-1}, (U - U_d, \tau - \tau_d))\\
    \tilde{\tau}_\psi \coloneqq \psi^\tau((Q_d, P_d)^{-1}, (U - U_d, \tau - \tau_d)). \label{eq:tilde_tau}
\end{align}
%where the input deltas $(U - U_d, \tau - \tau_d)$ are translated with the desired trajectory $(Q_d, P_d)$.
\end{subequations}

Then the error dynamics are
\begin{align}
    &\dot{Q}_E = Q_E \tilde{U}_\psi ; & &\dot{P}_E = \ad_{\tilde{U}_\psi}^\ast P_E + \tilde{\tau}_\psi \label{eq:error_dynamics}.
\end{align}
The error dynamics preserve the extended Euler-Poincare structure; that is, $\dot{e} = f(e, (\tilde{U}_\psi, \tilde{\tau}_\psi))$.
\end{theorem}

\begin{proof}
    First, note the following identities:
    \begin{align}
    \tD_1 \phi((Q_1, P_1),& (Q_2, P_2))[Q_1 U, W] \notag \\
        &= f(\phi((Q_1, P_1),(Q_2, P_2)), (U, W))  \label{eq:d1_ident}
    \end{align}
    and
    \begin{align}
            \ddt (Q, P)^{-1} = -f((Q, P)^{-1}, \psi((Q, P)^{-1}, (U, \tau))), \label{eq:d_inv_ident}
    \end{align}
    that can be proved with straight-forward computation.

    For \eqref{eq:d1_ident}, observe that
    \begin{align}
        \tD_1 \phi((&Q_1, P_1), (Q_2, P_2))[Q_1 V, W] \notag \\
                             %&= \frac{\td}{\td s} \bigg\rvert_{s=0} \phi((Q_1e^{Vs}, P_1 + sW), (Q_2, P_2)) \label{eq:d1_eq_1}\\
                             &= \frac{\td}{\td s} \bigg\rvert_{s=0} (Q_2 Q_1 e^{Vs} , \Ad_{Q_1 e^{Vs}}^\ast  P_2 + P_1 + sW) \label{eq:d1_eq_1}\\
                             &= (Q_2 Q_1 V, \ad_{V}^\ast  \Ad_{Q_1}^\ast  P_2 + W) \notag\\
                             &= f((Q_2 Q_1, \Ad_{Q_1}^\ast  P_2), (V, W)) \label{eq:d1_eq_2} \\
                            &= f(\phi((Q_1, P_2),(Q_2, P_2)), (V, W)) \notag,
    \end{align}
    where \eqref{eq:d1_eq_1} is the definition of the directional derivative and \eqref{eq:d1_eq_2} follows from the definition of $f$.

    For \eqref{eq:d_inv_ident}, note that
    \begin{align}
        \ddt (Q, P)^{-1} &= \ddt (Q^{-1}, -\Ad_{Q^{-1}}^\ast  P) \notag\\
        &= (-Q^{-1} QUQ^{-1}, \Ad_{Q^{-1}}^\ast  \ad_U^\ast  P \notag\\
         & \quad \quad - \Ad_{Q^{-1}}^\ast (\ad_U^\ast  P + \tau)) \label{eq:inverse_deriv_1}\\
        &= (-Q^{-1} QUQ^{-1}, \ad_{\Ad_Q U}^\ast  \Ad_{Q^{-1}}^\ast  P \notag\\
        & \quad \quad -\ad_{\Ad_Q U}^\ast  \Ad_{Q^{-1}}^\ast  P - \Ad_{Q^{-1}}^\ast  \tau) \label{eq:inverse_deriv_2}\\
        &= (-Q^{-1} \psi^U((Q^{-1}, -\Ad_{Q^{-1}}^\ast  P), (U, \tau)), \notag\\
        & \quad \quad  \ad_{\Ad_Q U}^\ast  \Ad_{Q^{-1}}^\ast  P \notag\\
        &\quad \quad  -\psi^\tau((Q^{-1}, -\Ad_{Q^{-1}}^\ast  P), (U, \tau))) \label{eq:inverse_deriv_3}\\
        &= -f((Q, P)^{-1}, \psi((Q, P)^{-1}, (U, \tau))) \label{eq:inverse_deriv_4},
    \end{align}
        where \eqref{eq:inverse_deriv_1} follows from \eqref{eq:Ad_inv_star_deriv} and the product rule, \eqref{eq:inverse_deriv_2} follows from \eqref{eq:Ad_inv_star_deriv}, \eqref{eq:inverse_deriv_3} follows from the definition of $\psi$ and \eqref{eq:inverse_deriv_4} follows from the definition of $f$.

    Finally, to prove \eqref{eq:error_dynamics} compute 
    \begin{align}
        \dot{e} &= D_1 \phi ((Q_d, P_d)^{-1}, (Q, P)) [\ddt (Q_d, P_d)^{-1}] \notag\\
                & \quad \quad + D_2 \phi((Q_d, P_d)^{-1}, (Q, P)) [\ddt (Q, P)] \label{eq:e_dot_1}\\
                &= D_1 \phi ((Q_d, P_d)^{-1}, (Q, P)) \notag\\
                &\quad \quad \quad \quad [-f((Q_d, P_d)^{-1}, \psi((Q_d, P_d)^{-1}, (U_d, \tau_d)))] \notag\\
                & \quad \quad + D\phi^P \phi((Q_d, P_d)^{-1}) [f((Q, P), (U, \tau))] \label{eq:e_dot_2}\\
                &= -f(\phi((Q_d, P_d)^{-1},(Q, P)), \psi((Q_d, P_d)^{-1}, (U_d, \tau_d))) \notag \\
                & \quad +f(\phi((Q_d, P_d)^{-1}, (Q, P)), \psi((Q_d, P_d)^{-1}), (U, \tau)) \label{eq:e_dot_3}\\
                % &= -f(e, \psi((Q_d, P_d)^{-1}, (U_d, \tau_d))) \notag\\
                % & \quad \quad +f(e, \psi((Q_d, P_d)^{-1}), (U, \tau)) \label{eq:e_dot_4}\\
                &= f(e, \psi((Q_d, P_d)^{-1}, (U - U_d, \tau - \tau_d))) \label{eq:e_dot_5} \\
                &= f(e, (\tilde{U}_\psi, \tilde{\tau}_\psi)) \notag
    \end{align}
    where \eqref{eq:e_dot_1} follows by applying the product and chain rules, \eqref{eq:e_dot_2} follows from identity \eqref{eq:d_inv_ident} and the definition of $f$,
    \eqref{eq:e_dot_3} follows from Lemmas \eqref{lem:lemma_eq_d2} and identity \eqref{eq:d1_ident} and \eqref{eq:e_dot_5} follows from the linearity of both $f$ and $\psi$ in the second argument.
\end{proof}

In the classical Euler-Poincare system the constraint $P = \tD l(U)$ couples velocity $U$ and the momentum $P$. 
If this is true for the trajectories $(Q(t),P(t))$ and $(Q_d(t), P_d(t))$ in Theorem \ref{thm:error_dyn} then it is of interest to consider whether a similar relationship holds for the error dynamics $(Q_E(t), P_E(t))$.
%Since the error system \eqref{eq:error_dynamics}  is
%also an Euler-Poincare system, there must be some ``error inertia'' relating the ``error momentum'' $P_E$ and the ``error velocity'' $\tilde{U}_\psi$.
Indeed, winding back definitions, we have
\begin{align*}
    P_E \coloneqq \Ad_{Q_d^{-1}}^\ast  (P - P_d) &=\Ad_{Q_d^{-1}}^\ast   \mathbb{I} \Ad_{Q_d}^{-1}  \Ad_{Q_d} (U - U_d)\\
        &= \Ad_{Q_d^{-1}}^\ast  \mathbb{I} \Ad_{Q_d}^{-1} \tilde{U}_\psi
\end{align*}
So a suitable time-varying ``inertia error'' $\bar{\mathbb{I}}$ is given by
\begin{align}
    \bar{\mathbb{I}} \coloneqq \Ad_{Q_d^{-1}}^\ast  \mathbb{I} \Ad_{Q_d^{-1}} \label{eq:I_bar},
\end{align}
so that $P_E = \bar{\mathbb{I}} \tilde{U}_\psi$.
This matrix can be interpreted as the the inertia matrix $\mathbb{I}$ expressed in the frame of the time-varying trajectory $Q_d$.
That is, the operator transforms a spatial velocity into the body frame with the rightmost $\Ad_{Q_d^{-1}}$, then applies $\mathbb{I}$, then transforms the resultant body-frame momentum back into the inertial frame with leftmost $\Ad_{Q_d^{-1}}^\ast$.

\subsection{Energy and Lyapunov analysis}
In the following development, we will assume that the linear map $\mathbb{I} : \gothg \to \gothg^\ast$ is symmetric (that is, once $\gothg^{**}$ and $\gothg$ are identified, then $\mathbb{I} = \mathbb{I}^\ast$).
    If the kinetic energy is used to define the Lagrangian $L$, then this assumption follows naturally.
    In this case, one can also show that $\mathbb{I}^{-1}$ and $\bar{\mathbb{I}}^{-1}$ are symmetric, facts that we will use but not prove.

    There is a natural choice of energy function associated with these error states, given by the product of the error momentum and velocity:
        \begin{align}
            \calL_E \coloneqq \frac{1}{2} P_E(\bar{\mathbb{I}}^{-1} P_E). \label{eq:err_energy}
        \end{align}

    \begin{lemma}
    \label{lem:energy_deriv}
        Let $\calL_E$ be defined as in \eqref{eq:err_energy}.
        Then the time derivative of $\calL_E$ is given by
        \begin{align*}
            \frac{\td}{\td t} &\calL_E = \tilde{\tau}_\psi (\bar{\mathbb{I}}^{-1} P_E) + \ad_{\Ad_{Q_d}U_d}^\ast P_E( \bar{\mathbb{I}}^{-1} P_E).
        \end{align*}
    \end{lemma}
    \begin{proof}
        \begin{align}
            \frac{\td}{\td t} &\calL_E = \frac{\td}{\td t} \frac{1}{2}  P_E(\bar{\mathbb{I}}^{-1} P_E) \notag\\
                                              &= \frac{1}{2}  \dot{P}_E (\bar{\mathbb{I}}^{-1} P_E) + \frac{1}{2}  P_E(\frac{\td}{\td t} (\bar{\mathbb{I}}^{-1}) P_E) + \frac{1}{2} P_E(\bar{\mathbb{I}}^{-1} \dot{P}_E) \notag \\
                                            %   &= \frac{1}{2} (\ad_{\tilde{U}_\psi}^\ast  P_E + \tilde{\tau}_\psi) (\bar{\mathbb{I}}^{-1} P_E) \notag\\
                                            %   & \quad \quad + \frac{1}{2} P_E((\ad_{\Ad_{Q_d}U_d} \bar{\mathbb{I}}^{-1}  + \bar{\mathbb{I}}^{-1} \ad_{\Ad_{Q_d}U_d}^\ast ) P_E) \label{eq:energy_0}\\
                                            %   & \quad \quad + \frac{1}{2} P_E(\bar{\mathbb{I}}^{-1}(\ad_{\tilde{U}_\psi}^\ast  P_E + \tilde{\tau}_\psi)) \notag\\
                                              &= \frac{1}{2} \ad_{\tilde{U}_\psi}^\ast  P_E (\bar{\mathbb{I}}^{-1} P_E) + \frac{1}{2} \tilde{\tau}_\psi (\bar{\mathbb{I}}^{-1} P_E) \notag\\
                                              & \quad \quad + \frac{1}{2} P_E(\ad_{\Ad_{Q_d}U_d} \bar{\mathbb{I}}^{-1} P_E)\notag\\
                                              & \quad \quad + \frac{1}{2} P_E(\bar{\mathbb{I}}^{-1} \ad_{\Ad_{Q_d}U_d}^\ast  P_E) \label{eq:energy_0}\\
                                              & \quad \quad + \frac{1}{2} \ad_{\tilde{U}_\psi}^\ast  P_E(\bar{\mathbb{I}}^{-1} P_E) + \frac{1}{2} P_E(\bar{\mathbb{I}}^{-1} \tilde{\tau}_\psi) \label{eq:energy_1}\\
                                            %   &=  \frac{1}{2} P_E (\ad_{\tilde{U}_\psi} \bar{\mathbb{I}}^{-1} P_E) + \frac{1}{2} \tilde{\tau}_\psi (\bar{\mathbb{I}}^{-1} P_E) \notag\\
                                            %   & \quad \quad + \frac{1}{2} P_E(\ad_{\Ad_{Q_d}U_d} \bar{\mathbb{I}}^{-1} P_E) \notag  \\
                                            %   & \quad \quad + \frac{1}{2} P_E(\bar{\mathbb{I}}^{-1} \ad_{\Ad_{Q_d}U_d}^\ast  P_E) \notag\\
                                            %   & \quad \quad + \frac{1}{2} P_E(\ad_{\tilde{U}_\psi } \bar{\mathbb{I}}^{-1} P_E) + \frac{1}{2} P_E(\bar{\mathbb{I}}^{-1} \tilde{\tau}_\psi), \notag\\
                                            %   &= \frac{1}{2} \tilde{\tau}_\psi (\bar{\mathbb{I}}^{-1} P_E) + \frac{1}{2} P_E(\bar{\mathbb{I}}^{-1} \tilde{\tau}_\psi) \notag\\
                                            %   & \quad \quad + \frac{1}{2} P_E(\ad_{\Ad_{Q_d}U_d} \bar{\mathbb{I}}^{-1} P_E) \notag \\
                                            %   & \quad \quad + \frac{1}{2} P_E(\bar{\mathbb{I}}^{-1} \ad_{\Ad_{Q_d}U_d}^\ast  P_E) \notag \\
                                              &= \tilde{\tau}_\psi (\bar{\mathbb{I}}^{-1} P_E) + \ad_{\Ad_{Q_d}U_d}^\ast P_E( \bar{\mathbb{I}}^{-1} P_E) \label{eq:energy_3},
         \end{align}
         where we have used the fact that $\tilde{U}_\psi = \bar{\mathbb{I}}^{-1} P_E$, so $\ad_{\tilde{U}_\psi}(\bar{\mathbb{I}}^{-1} P_E) = 0$, \eqref{eq:energy_0} follows from Proposition \ref{prop:I_bar_dot} and \eqref{eq:energy_1} and \eqref{eq:energy_3} follow from the symmetry of $\bar{\mathbb{I}}^{-1}$.
        \end{proof}

    %But then what
    %So we probably need to just look at the isomorphisms of $\gothg$ and $\gothg^\ast  with $\R^3$ and take the corresponding dot product}

    %Define an energy term $\langle p_E, I^{-1}p_E \rangle$

    % Note:
    % \begin{align*}
    %     p_E &= \Ad_{Q_d^{-1}}^\ast  (p - P_d)\\
    %         & = \Ad_{Q_d^{-1}}^\ast  I \tilde{u}\\
    %         & = \Ad_{Q_d^{-1}}^\ast  I \Ad_{Q_d^{-1}} \tilde{u}\\
    % \end{align*}

    % Define $\bar{\mathbb{I}} \coloneqq \Ad_{Q_d^{-1}}^\ast  I \Ad_{Q_d^{-1}}$.

    % $\bar{\mathbb{I}}^{-1} = \Ad_{Q_d} I^{-1} \Ad_{Q_d}^\ast $

    % $relate to fibre deriv$

    % $\cal{L} = l(p_E, p_E)?$

    % Define inner product so that L(a) = <a,a>

    % is I equal to fibre derivative?

    % Will eventually use
    % $\cal{L} = H(p_E) + <I - E, I - E>$

% \subsection{Lyapunov stuff}

\begin{theorem}
    \label{thm:stability}
Let the error trajectory $e(t)\in \grpG^\ast _{\ltimes}$ be defined by \eqref{eq:e_def} with dynamics \eqref{eq:error_dynamics}.
Let the desired configuration trajectory $Q_d(t)$ and its inverse $Q_d^{-1}(t)$ be bounded.
Define the input $\tilde{\tau}_\psi$ by
\begin{align}
    \tilde{\tau}_\psi = - k_v K(\tilde{U}_\psi) -\ad_{\Ad_{Q_d}U_d}^\ast P_E -k_p K \left(\mathbb{P}_\gothg (Q_E^\top(Q_E - I)) \right), \label{eq:control_law}
\end{align}
where $k_p, k_v \in \R^+$ are fixed gains, $K : \gothg \to \gothg^\ast$ is the linear map defined by $K(V) = \langle V, \cdot \rangle$ and $\mathbb{P}_\gothg : \R^{m \times m} \to \gothg$ is the projection operator that takes an $m \times m$ matrix and returns the skew-symmetric component.
Then $e(t)$ converges to the equilibrium set $\left\{ (Q_E, P_E) : \mathbb{P}_\gothg (Q_E^\top(I - Q_E)) = 0, P_E = 0 \right\}$.
In particular, the equilibrium set contains the identity as an isolated stable equilibrium point, so the choice of input renders the error system \eqref{eq:error_dynamics} locally asymptotically stable at identity.
\end{theorem}

\begin{proof}
Define the Lyapunov function
\begin{align}
    \calL = \frac{1}{2} P_E(\bar{\mathbb{I}}^{-1} P_E) + \frac{k_p}{2}\langle Q_E - \Id, Q_E - \Id \rangle \label{eq:lyap},
\end{align}
which is clearly positive definite.
The time derivative $\dot{\calL}$ is given by
\begin{align}
    \dot{\calL} &= \tilde{\tau}_\psi (\bar{\mathbb{I}}^{-1} P_E) + P_E(\ad_{\Ad_{Q_d}U_d} \bar{\mathbb{I}}^{-1} P_E) \label{eq:lyap_deriv_1}\\
                    & \quad \quad \quad \quad + k_p \langle Q_E - \Id, Q_E \tilde{U}_\psi \rangle \label{eq:lyap_deriv_2},
\end{align}
where the term \eqref{eq:lyap_deriv_1} comes from Lemma \ref{lem:energy_deriv} and the term \eqref{eq:lyap_deriv_2} comes from \eqref{eq:error_dynamics}.
Substituting the control law \eqref{eq:control_law}, we find
\begin{align}
    \dot{\calL} &= (- k_v K(\tilde{U}_\psi)) (\bar{\mathbb{I}}^{-1} P_E) \notag\\
                  & \quad \quad -\ad_{\Ad_{Q_d}U_d}^\ast P_E(\bar{\mathbb{I}}^{-1} P_E) +  \ad_{\Ad_{Q_d}U_d}^\ast P_E( \bar{\mathbb{I}}^{-1} P_E) \notag \\
                  & \quad \quad -k_p K \left(\mathbb{P}_\gothg (Q_E^\top(Q_E - \Id)) \right)(\bar{\mathbb{I}}^{-1} P_E) \notag\\
                  & \quad \quad + k_p \langle Q_E - \Id, Q_E \tilde{U}_\psi \rangle \notag\\
                  &= -k_v \langle \tilde{U}_\psi, \tilde{U}_\psi \rangle \label{eq:l_dot_1}\\
                  & \quad \quad -k_p \langle \mathbb{P}_\gothg (Q_E^\top(Q_E - \Id)), \tilde{U}_\psi \rangle + k_p \langle Q_E - \Id, Q_E \tilde{U}_\psi \rangle \label{eq:l_dot_2}\\
                  &= -k_v \lVert\tilde{U}_\psi \rVert^2 \notag\\
                  & \quad \quad -k_p\langle Q_E - \Id, Q_E \tilde{U}_\psi \rangle  + k_p \langle Q_E - \Id, Q_E \tilde{U}_\psi \rangle \label{eq:l_dot_3}\\
                  &= -k_v \lVert\tilde{U}_\psi \rVert^2 \notag,
\end{align}
where \eqref{eq:l_dot_1} and \eqref{eq:l_dot_2} follow from the definition of $K$ and $\tilde{U}_\psi$, and \eqref{eq:l_dot_3} follows from the properties of the Frobenius inner product and the fact that $\langle \mathbb{P}_\gothg A, V \rangle = \langle A, V \rangle$ for all $A \in \R^{m \times m}$ and $V \in \gothg$.
Hence $\dot{\calL}$ is negative semi-definite.

We will invoke Barbalat's lemma \mbox{(\cite{slotine1991applied})} to prove the theorem claim.
The Lyapunov function $\calL(t)$ is bounded above by $\calL(0)$ and is non-increasing and so $\lim_{t \to \infty} \calL(t) = c \geq 0$. 
In particular, as both sides of \eqref{eq:lyap} are strictly positive, this implies that $Q_E$ and $P_E$ (and hence $\tilde{U}_\psi$) are bounded.
We now show that this implies that $\dot{\tilde{U}}_\psi$ (and $P_E$) are bounded.
One can compute
\begin{align*}
    \ddt \tilde{U}_\psi &= (\ddt \bar{\mathbb{I}}^{-1}) P_E + \bar{\mathbb{I}}^{-1} \dot{P}_E\\
                        &= \ad_{\Ad_{Q_d}U_d} \bar{\mathbb{I}}^{-1}P_E + \bar{\mathbb{I}}^{-1} \ad_{\Ad_{Q_d}U_d}^*P_E\\
                        & \quad \quad + \bar{\mathbb{I}}^{-1}\ad_{\tilde{U}_\psi}^* P_E + \bar{\mathbb{I}}^{-1} \tilde{\tau}_\psi.
\end{align*}
The first three terms are linear functions of bounded variables, so are also bounded.
Referring to the definition of $\tilde{\tau}_\psi$ \eqref{eq:control_law}, the same also applies to this term, so $\dot{\tilde{U}}_\psi$ (and $P_E$) are bounded.

From here, one can compute $\ddot{\calL} =  -2 k_v \langle \dot{\tilde{U}}_\psi, \tilde{U}_\psi \rangle$.
Because $\tilde{U}_\psi$ and $\dot{\tilde{U}}_\psi$  bounded, then $\ddot{\calL}$ is bounded by Cauchy-Schwarz.
Thus, by Barbalat's lemma, $\dot{\calL} \to 0$, $\tilde{U}_\psi \to 0$ and $P_E \to 0$.

Now, because $\dot{P}_E$ is a linear function of $P_E$ terms, it is bounded and so $P_E$ is uniformly continuous.
We observe that $\dot{P}_E$ is therefore a composition of uniformly continuous functions, so is also uniformly continuous.
Then by Barbalat's lemma, because $P_E \to 0$, we must also have $\dot{P}_E \to 0$.
In particular, referring to the definition \eqref{eq:error_dynamics}, then $\tilde{\tau}_\psi \to 0$.
Referring to \eqref{eq:control_law}, the terms $K(\tilde{U}_\psi)$ and $\ad_{\Ad_{Q_d} U_d}$ both go to zero, so we must also have that $\mathbb{P}_\gothg (Q_E^\top(I - Q_E)) \to 0$.
Therefore $\tilde{\tau} \to 0$ and hence $Q_E$ approaches the equilibrium set.

Clearly, if $Q_E = \Id$, then $\mathbb{P}_\gothg (Q_E^\top(I - Q_E)) = 0$, so $(\Id,0)$ is an isolated equilibrium point.
Additionally, $(Q_E,P_E) = (\Id,0)$ is the global minimum of $\calL$, so it is stable.
\end{proof}

%%%%%%%%%%%%%%%%%%%%%%%%%%%%%%%%%%%%%%%%%%%%%%%%%%%%%%%%%%%%%%%%%%%%%%%%%%%%%%%%%%%%%%%%%%%%%%%%%%
\section{Example - controller design on $\SO(3)$}
%%%%%%%%%%%%%%%%%%%%%%%%%%%%%%%%%%%%%%%%%%%%%%%%%%%%%%%%%%%%%%%%%%%%%%%%%%%%%%%%%%%%%%%%%%%%%%%%%%

We now proceed with an example of fully-actuated attitude control.
This is a standard example with known solutions (\cite{bayadiAlmostGlobalAttitude2014}) and so serves as a suitable benchmark to compare our approach.
We identify each of the components (the semi-direct Lie group structure, dynamics, error structure, error dynamics and Lyapunov function) required for our approach.

\subsection{Attitude dynamics and Lie group structure}
We quickly recount the standard construction of attitude dynamics.
For a more detailed construction, we refer the reader to (\cite{marsden1994introduction}).
Attitude is typically modelled on the Lie group $\SO(3)$.
The dynamics on $\SO(3)$ are well-studied, but are usually approached with the velocity and momentum modelled in $\R^3$.
With the identification of $\so^\ast(3) \cong \R^3$, and the use of the identities given in \eqref{eq:r3_ident}, the Euler-Poincare equations \eqref{eq:dynamics} on $\SO(3)$ become the familiar equations (\cite{marsden1994introduction})
\begin{subequations}\label{eq:so3_dynamics}
\begin{align}
\dot{R} &= R \Omega^\times\\
\mathbb{I}\dot{\Omega} &= -\Omega \times \mathbb{I}\Omega + \tau,
\end{align}
\end{subequations}
where $\mathbb{I}$ is the inertia map acting on $\R^3$, $\Omega \in \R^3$ is the body-frame angular velocity and  $\tau \in \R^3$ is a torque input.
In the following equivariant analysis, $\Omega$ will be treated as an input and the physical constraint relating $\Omega$ and the momentum $\mathbb{I} \Omega$ will be reintroduced during control selection.

\subsection{Semidirect group structure and equivariance}

Recall the general semidirect group structure we are interested in \eqref{eq:semidirect}.
Using the identities \eqref{eq:r3_ident}, the corresponding structure on the space $\SO(3) \times \R^3$ is given by
\begin{align*}
    (R_1, p_1)(R_2, p_2) = (R_1 R_2, R_2^\top p_1 + p_2).
\end{align*}
%With the identification of $\so(3)$ and $\R^3$, we have $\ad_{v} u \cong v \times u$, $\ad_{v}^\ast P u = p^\top (v \times u) = u^\top (p \times v)$.
%So $\ad_{v}^\ast P \cong - v \times p$.

% As such, the set $\SO(3) \times \so^\ast(3)$ with group product
% \begin{align*}
%     (R_1, p_1)(R_2, p_2) = (R_1 R_2, \Ad_{R_2}^\ast  P_1 + P_2)
% \end{align*}
% is isomorphic to $\SO(3) \times \R^3$ with the group product
% \begin{align*}
%     (R_1, p_1)(R_2, p_2) = (R_1 R_2, R_2^\top p_1 + p_2).
% \end{align*}

%(R_1^\top, p_2 = - R_2^\top p_1)

With this identification, the group action defined in \eqref{eq:phi_def} is equivalent to the transformation
\begin{align*}
    \phi((R_1,p_2), (R_2, p_2)) \coloneqq \left( R_2 R_1, R_1^\top p_2 + p_1  \right).
\end{align*}

Similarly, the input action defined in \eqref{eq:psi_def} is equivalent to the transformation
\begin{align*}
    \psi((R,& p), (U, \tau)) \coloneqq \left( R^\top \Omega^\times, \quad R^\top \tau + (R^\top \Omega) \times p \right).
\end{align*}

\subsection{Error structure and error dynamics}

Comparing the error definitions \eqref{eq:e_def} to the corresponding identities in $\R^3$ \eqref{eq:r3_ident}, we have
\begin{align*}
    &R_E = R R_d^\top, &p_E = R_d (p - p_d)
\end{align*}
and
\begin{align*}
    \tilde{\Omega}_\psi &\coloneqq R_d (\Omega - \Omega_d)\\
    \tilde{p}_\psi &\coloneqq R_d (\tau - \tau_d) - (R_d (\Omega - \Omega_d)) \times R_d p_d.
\end{align*}
By Theorem \ref{thm:error_dyn}, the associated error dynamics are given by
\begin{subequations}\label{eq:so3_err_dynamics}
\begin{align}
    \dot{R}_E &= R_E \tilde{\Omega}_\psi^\times\\
    \dot{p}_E &= -\tilde{\Omega}_\psi \times p_E + \tilde{\tau}_\psi.
\end{align}
\end{subequations}

The inertia error map $\bar{\mathbb{I}}$ is given by $\bar{\mathbb{I}} \coloneqq R_d \mathbb{I} R_d^\top$, so that $p_E = \bar{\mathbb{I}} \tilde{\Omega}_\psi$.

\subsection{Lyapunov function and control}
%\tilde{\tau}_\psi = - k_v K(\tilde{U}_\psi) -\ad_{\Ad_{Q_d}U_d}^\ast P_E + K \left(\mathbb{P}_\gothg (Q_E^\top(I - Q_E)) \right),
Referring to the general control law \eqref{eq:control_law} and using the corresponding identities in $\R^3$ \eqref{eq:r3_ident} and \eqref{eq:so3_proj}, the control on $\R^3$ is given by
\begin{align}
    \tilde{\tau}_\psi = - k_v \tilde{\Omega}_\psi + (R_d \Omega_d) \times p_E  - \frac{k_p}{2} (R_E - R_E^\top)^\vee. \label{eq:so3_input}
\end{align}
There are some additional characterisations of the equilibria for this example, which we capture in the following proposition.
\begin{proposition}
    Let $e(t) \coloneqq (R_E(t), p_E(t))$ evolve according to \eqref{eq:so3_err_dynamics} with input given by \eqref{eq:so3_input}.
    The stable and unstable equilibria of $e$ are given by
    $\left\{ (R_E, p_E) : R_E = \Id, p_E = 0 \right\}$ and $\left\{ (R_E, p_E) : \tr(R_E)   = -1, p_E = 0 \right\}$, respectively.
\end{proposition}
\begin{proof}
    From Theorem \ref{thm:stability}, the set of equilibria are given by
    \begin{align*}
        \calE = \left\{ (R_E, p_E) : \frac{1}{2}(R_E - R_E^\top) = 0, p_E = 0 \right\}.
    \end{align*}
    Clearly, $(I, 0) \in \calE$.
    The stability of the point $(\Id, 0)$ follows from Theorem \ref{thm:stability}.
    The remaining equilibria correspond to the set of symmetric rotation matrices $R_E = R_E^\top$ (excluding identity), and so correspond to rotations about some axis with angle $\pi$.
    Thus, the remaining equilibria have eigenvalues $\{1, -1, 1\}$ (\cite{baker2000introduction}[page 31]) and so by the spectral theorem (\cite{hawkins1975cauchy}), the remaining equilibria correspond to the set of matrices given by
    \begin{align*}
        \left\{ R \Lambda R^\top: R \in \SO(3), \Lambda = \diag(1, -1, -1) \right\}.
    \end{align*}
    To show that this set is unstable, it suffices to show that for an arbitrary point in this set, all neighborhoods of this set contain another point for which the Lyapunov function is smaller.
    Consider such an arbitrary point $(R \Lambda R^\top, 0)$, consider the curve
    \begin{align*}
        \left\{Q(s) = R \Lambda R^\top \exp(s(R e_1)^\times) : s \in (-\varepsilon, \varepsilon), \varepsilon \in \R^+ \right\},
    \end{align*}
    and apply a second-order expansion of $\calL$ around $s$:
    \begin{align*}
        \calL(&Q(s), 0)\\
                         &= \tr((\Id - Q(s))^\top(\Id - Q(s)))\\
                         &= 2\tr(\Id -Q(s))\\
                         &= 2 \tr(\Id - R \Lambda R^\top \exp(s(R e_1)^\times))\\
                         &\approx 2 \tr(\Id - R \Lambda R^\top (\Id + s(R e_1)^\times + \frac{s^2}{2} (Re_1)^\times (R e_1)^\times))\\
                         &= 8 - \tr(2sR \Lambda R^\top (R e_1)^\times - s^2 R \Lambda R^\top (Re_1)^\times (R e_1)^\times)\\
                         &= 8 - s^2 \tr(\Lambda e_1^\times e_1^\times)\\
                         &= 8 - s^2 (e_1^\top \Lambda e_1 - \tr(\Lambda)) = 8 - 2s^2 \leq 8 = \calL(Q(s), 0).
    \end{align*}
    Therefore all points in this set are unstable.
\end{proof}

\subsection{Simulation}
\label{sec:simulation}
To verify the approach, we simulate the proposed $\SO(3)$ controller to track a sample trajectory.
The trajectory to track is defined by $R_{d_0} = \Id$, $\Omega_{d_0} = (0, 0, 0)$, $\tau_d(t) = (\cos(t), \sin(t), \sin(t)\cos(t))$.
The system trajectory is perturbed by setting the initial conditions to
\begin{align*}
   &R_0 = \exp((\pi - 0.1) (0.8, 0.6, 0.0)^\times), &\Omega_0 = (4, -3, 2).
\end{align*}
The simulation is run with a timestep of $0.01$ and the gains are chosen to be $k_v = 0.5$, $k_p = 1.0$.
In this simulation, the constraints $p = \mathbb{I} \Omega, p_d = \mathbb{I} \Omega_d$ are enforced in the construction of $\tilde{\tau}_\psi$, which requires $\Omega, \Omega_d$ to define the $\tilde{\Omega}_\psi$ term.
After $\tilde{\tau}_\psi$ is constructed, the physical input $\tau$ can be selected by the definition \eqref{eq:tilde_tau}, so that
\begin{align*}
    \tau = \tau_d + R_d^\top \tilde{\tau}_\psi + (\Omega - \Omega_d) \times p_d.
\end{align*}
We compute $\lVert R_E - \Id \rVert$ using the Frobenius norm ($\lVert R_E - \Id \rVert = \sqrt{\langle R_E - \Id, R_E - \Id \rangle}$) and $\lVert p_E \rVert$ using the standard $\R^3$ norm and display the results of this simulation in Figure \ref{fig:simulation}.

\begin{figure}[!tb]
    \includegraphics[width=0.7\linewidth]{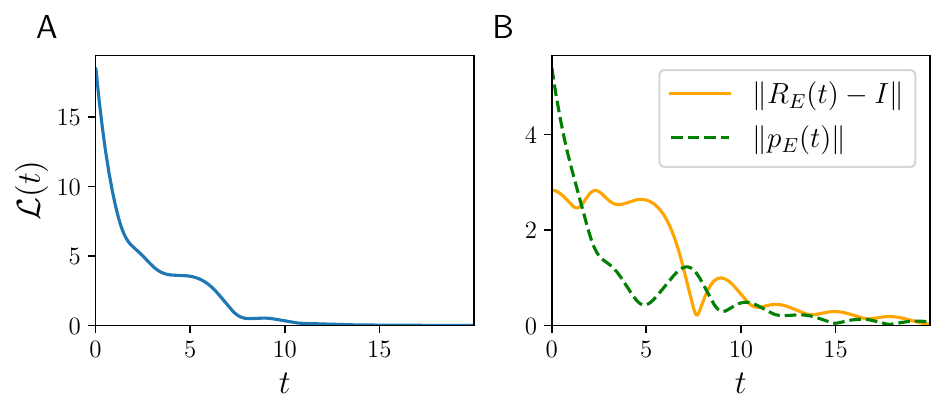}
    \centering
    \caption{A: The lyapunov function $\calL(t)$ vs time for the simulation described in \S \eqref{sec:simulation}. B: The magnitudes of the individial errors $\lVert R_E \rVert$ and $\lVert p_E \rVert$ vs time.}
    \label{fig:simulation}
\end{figure}

\section{Conclusion}

In this paper, we have approached the tracking task for fully-actuated mechanical systems whose configuration space is a matrix Lie group.
In this case, the system evolution is determined by the Euler-Poincare equations.
We have extended the the natural symmetry on the configuration space to the phase space and proposed an input action to show that the extended Euler-Poincare equations are equivariant with respect to this symmetry.
A key step in this approach was characterising the velocity state as an input.
We have shown that this choice of equivariant symmetry leads to error dynamics that also satisfy extended Euler-Poincare dynamics. 
We propose a simple PD type controller and prove the (local) asymptotic convergence of the error dynamics for this control scheme.
Finally, we have verified the effectiveness of this control approach in simulation for attitude control, where the control results in almost-global asymptotic stability of the error dynamics.
%%%%%%%%%%%%%%%%%%%%%%%%%%%%%%%%%%%%%%%%%%%%%%%%%%%%%%%%%%%%%%%%%%%%%%%%%%%%%%%%%%%%%%%%%%%%%%%%%%
\bibliographystyle{plainnat}
\bibliography{references}             % bib file to produce the bibliography
\section{Appendix}

\begin{proposition}
    \label{prop:I_bar_dot}
    Let the operator $\bar{\mathbb{I}}$ be defined as in \eqref{eq:I_bar}.
    The time derivative of $\bar{\mathbb{I}}^{-1}$ is given by
    \begin{align*}
        \frac{\td}{\td t} &\bar{\mathbb{I}}^{-1} = \ad_{\Ad_{Q_d}U_d} \bar{\mathbb{I}}^{-1}  + \bar{\mathbb{I}}^{-1} \ad_{\Ad_{Q_d}U_d}^\ast
    \end{align*}
\end{proposition}
\begin{proof}
    By direct computation, we have
    \begin{align}
        \frac{\td}{\td t} &\bar{\mathbb{I}}^{-1} \notag\\
                &= (\frac{\td}{\td t} \Ad_{Q_d}) \mathbb{I}^{-1} \Ad_{Q_d}^\ast  + \Ad_{Q_d} \mathbb{I}^{-1} (\frac{\td}{\td t} \Ad_{Q_d}^\ast ) \label{eq:I_deriv_1}\\
                &=  \ad_{\Ad_{Q_d}U_d} \Ad_{Q_d} \mathbb{I}^{-1} \Ad_{Q_d}^\ast  \notag \\
                & \quad \quad + \Ad_{Q_d} \mathbb{I}^{-1}  \Ad_{Q_d}^\ast  \ad_{\Ad_{Q_d}U_d}^\ast  \label{eq:I_deriv_2}\\
                &=  \ad_{\Ad_{Q_d}U_d} \bar{\mathbb{I}}^{-1}  + \bar{\mathbb{I}}^{-1} \ad_{\Ad_{Q_d}U_d}^\ast  \label{eq:I_deriv_3},
    \end{align}
    where \eqref{eq:I_deriv_1} follows from the definition of $\bar{\mathbb{I}}$ and the product rule, \eqref{eq:I_deriv_2} follows from identities \eqref{eq:Ad_deriv} and \eqref{eq:Ad_star_deriv}, and \eqref{eq:I_deriv_3} follows from the definition of $\bar{\mathbb{I}}$.
\end{proof}
% with bibtex (preferred)

%\begin{thebibliography}{xx}  % you can also add the bibliography by hand

%\bibitem[Able(1956)]{Abl:56}
%B.C. Able.
%\newblock Nucleic acid content of microscope.
%\newblock \emph{Nature}, 135:\penalty0 7--9, 1956.

%\bibitem[Able et~al.(1954)Able, Tagg, and Rush]{AbTaRu:54}
%B.C. Able, R.A. Tagg, and M.~Rush.
%\newblock Enzyme-catalyzed cellular transanimations.
%\newblock In A.F. Round, editor, \emph{Advances in Enzymology}, volume~2, pages
%  125--247. Academic Press, New York, 3rd edition, 1954.

%\bibitem[Keohane(1958)]{Keo:58}
%R.~Keohane.
%\newblock \emph{Power and Interdependence: World Politics in Transitions}.
%\newblock Little, Brown \& Co., Boston, 1958.

%\bibitem[Powers(1985)]{Pow:85}
%T.~Powers.
%\newblock Is there a way out?
%\newblock \emph{Harpers}, pages 35--47, June 1985.

%\bibitem[Soukhanov(1992)]{Heritage:92}
%A.~H. Soukhanov, editor.
%\newblock \emph{{The American Heritage. Dictionary of the American Language}}.
%\newblock Houghton Mifflin Company, 1992.

%\end{thebibliography}
% in the appendices.

\end{document}